\documentclass[a4paper, 11pt]{article}

\usepackage[margin=1in]{geometry}

\usepackage{setspace}
\usepackage{amsmath, amssymb, amsthm, mathtools}	% required for `\align*' (yatex added)
\usepackage{enumitem}

\newtheorem{assumption}{Assumption}
\newtheorem{theorem}{Theorem}
\newtheorem{lemma}{Lemma}

\makeatletter
\newcommand{\leqnomode}{\tagsleft@true}
\newcommand{\reqnomode}{\tagsleft@false}
\makeatother

\newcommand{\abs}[1]{\left\lvert #1 \right\rvert}
\DeclarePairedDelimiter\norm\lVert\rVert

\newcommand{\sumtT}{\sum_{t=1}^T}
\newcommand{\E}{\mathbb{E}}

\DeclareMathOperator{\sgn}{sign}
\DeclareMathOperator{\cov}{cov}

\usepackage[style = authoryear, uniquename = false, hyperref = true]{biblatex}
\title{Convergence rate of estimators of clustered panel models with misclassification\thanks{The research presented in this paper started as  Appendix G of a previous version of \textcite{dzemskiokui2019confidence} and has been substantially extended. The relevant results are removed from \textcite{dzemskiokui2019confidence}. Okui gratefully acknowledges the financial support of the School of Social Sciences and a New Faculty Startup Grant at Seoul National University and from the Housing and Commercial Bank Economic Research Fund in the Institute of Economic Research at Seoul National University. Dzemski gratefully acknowledges financial support from Jan Wallanders och Tom Hedelius samt Tore Browaldhs stiftelse grant P19-0079. A part of this research was done while Okui was at NYU Shanghai.}}

\author{Andreas Dzemski\thanks{Department of Economics, University of Gothenburg,
P.O. Box 640, SE-405 30 Gothenburg, Sweden.
Email: andreas.dzemski@economics.gu.se
} \ 
and Ryo Okui\thanks{Department of Economics and the Institute of Economic Research, Seoul National University,
Building 16, 1 Gwanak-ro, Gwanak-gu, Seoul, 08826, South Korea.
Department of Economics, University of Gothenburg,
P.O. Box 640, SE-405 30 Gothenburg, Sweden.
		Email: okuiryo@snu.ac.kr;
}}

\date{\today}

\allowdisplaybreaks

\begin{document}

\onehalfspacing

\maketitle
	
\begin{abstract}
	We study \emph{kmeans} clustering estimation of panel data models with a latent group structure and $N$ units and $T$ time periods under long panel asymptotics. We show that the group-specific coefficients can be estimated at the parametric root $NT$ rate even if error variances diverge as $T \to \infty$ and some units are asymptotically misclassified. This limit case approximates empirically relevant settings and is not covered by existing asymptotic results.
\par 
\vspace{4mm}
\textit{Keywords:} Panel data, latent grouped structure, clustering, kmeans, convergence rate, misclassification.
\par 
\vspace{1mm}
\textit{JEL codes:} C23, C33, C38

\textit{Mathematical Subjects Classification (2010):} 62H30, 62H12

\textit{Declarations of interest}: none

\end{abstract}

\section{Introduction}

Panel models can account for unobserved heterogeneity by dividing units into a finite number of latent groups and allowing a unit's coefficients to be group-specific \parencite{bonhomme2015grouped,su2016identifying, vogt2015classification, wang2016homogeneity, okuiwang2018}. Estimators of such models simultaneously estimate group memberships and group-specific coefficients. For example, \textcite{bonhomme2015grouped} propose a \textit{kmeans}-type estimator and \textcite{su2016identifying} propose the CLasso estimator that is based on solving a penalized regression program. These two and other related estimators are justified under a long panel asymptotic framework that sends both the number of units $N$ and the number of time periods $T$ to infinity. Existing theoretical results show that coefficients that are group-specific and time invariant can be estimated at a root $NT$ rate, i.e., at the parametric rate. In this paper we show that the parametric rate can be obtained even if some units have a positive probability of being misclassified in the limit. 
This limit case is highly relevant in practice since it is common to misclassify at least some units in empirical applications \parencite{BonhommeDiscretize}. However, existing results do not apply in such settings.

Existing asymptotic results for linear panel models assume that the variance of the error term is universally bounded. From this assumption, it can then be shown that group memberships can be estimated uniformly consistently, i.e., the probability of misclassifying one or more units vanishes as $N, T \to \infty$. This implies that the rate at which group-specific coefficients can be estimated is the same as under a known group structure and is therefore equal to the parametric rate. 

However, the assumption of a universal bound on the variance of the error term may not reflect real circumstances. It implies that the asymptotic limit as $T \to \infty$ prescribes that, for each unit, the level of statistical noise is negligible when compared to the number of observed time periods. This is not characteristic of typical empirical applications. The number of observed time periods is often rather small and, at least for some units, statistical noise plays an important role in determining the outcome.

In this paper, we extend previous theoretical results to a heteroscedastic setting in which units are endowed with unit-specific error variances $\sigma_1^2, \dotsc, \sigma_N^2$. A unit $i$ with small $\sigma_i$ is easy to classify, whereas a unit $i$ with large $\sigma_i$ is difficult to classify. The individual error variances may depend on $N$ and $T$ and may diverge as $T \to \infty$. We expect our asymptotic framework to be a more faithful approximation of the finite sample behavior of the estimators than the conventional framework.

For \emph{kmeans}-estimation, we show that uniform consistency of group memberships holds provided that the unit-specific error variances do not diverge too fast. Units $i$ for which $\sigma_i$ diverges too fast are potentially misclassified in the limit. However, if the proportion of such potentially misclassified units is sufficiently small then it is still possible to estimate the group-specific coefficients at a root $NT$ rate.

\textcite{pollard1981strong,pollard1982central,bonhomme2015grouped} consider panel models with fixed $T$ and estimate cluster-specific coefficients. They show that the cluster-specific coefficients converge to a pseudo-true value at rate root $N$ even though units are misclassified in the limit with positive probability. Their setting and results are distinct from ours. We consider long panel asymptotics under which true rather than pseudo-true cluster-specific coefficients can be identified and estimated at a root $NT$ rate.

We prove our results for a simple linear panel model with group-specific intercepts and focus on estimation by least squares (equivalent to \emph{kmeans}). By focusing on this simple model we are able to derive our results under interpretable and intuitive conditions on the structure of heteroscedasticity. While we think that our argument can be extended to more general regression models with group-specific coefficients, we believe that such an exercise would impose more involved assumptions and would not be as instructive about the mechanisms that allow root $NT$-consistency to arise despite of diverging error variances and possibly misclassified units.

\textcite{bonhomme2015grouped} conduct a simulation experiment that is calibrated to their empirical application. They find that the group-specific coefficients are estimated precisely, even though it is likely that one or more units are misclassified. Existing theoretical results about the rate of consistency of the group-specific coefficients cannot explain this phenomenon as they do not apply in the presence of misclassification. We fill this gap in the literature by showing that uniform consistency is sufficient but not necessary for precise estimation of the group-specific coefficients.

\section{Setting}

The units $i = 1, \dotsc, N$ are partitioned into $G$ groups. The set of all groups is $\mathbb{G} = \{1, \dotsc, G\}$ and unit $i$ belongs to group $g_i^0 \in \mathbb{G}$. 
For units in group $g \in \mathbb{G}$ the mean outcome in each period is given by $\mu_g$.
At time $t = 1, \dotsc, T$ we observe the scalar outcome $y_{it}$ generated by 
\begin{align*}
	y_{it} = \mu_{g_i^0} + \sigma_i v_{it},
\end{align*}
where $v_{it}$ is a noise term with variance one. 
Let $\Gamma$ denote the space of possible group assignments $\mathbf{g} = (g_1, \dotsc, g_N)$ and let $\mathcal{M}$ denote the space of possible group-specific means $\boldsymbol{\mu} = (\mu_1, \dotsc, \mu_G)$. The true group assignment $\mathbf{g}^0 \in \Gamma$ and the true group-specific mean $\boldsymbol{\mu}^0 \in \mathcal{M}$ are unknown parameters and are estimated. 

We consider \textit{kmeans}-type estimation as suggested in \textcite{bonhomme2015grouped}.
The objective function for estimation is defined on $\Gamma \times \mathcal{M}$ and is given by 
\begin{align*}
	Q_{N, T}(\mathbf{g}, \boldsymbol{\mu}) = 
	\frac{1}{NT} \sum_{i = 1}^N \sumtT
	\left(
		y_{it} - \mu_{g_i}
	\right)^2. 
\end{align*}
The estimator is defined as $(\hat{\boldsymbol{\mu}}, \hat{\mathbf{g}})= \arg \min_{ \boldsymbol{\mu} \in \mathcal{M}, \boldsymbol{g} \in \Gamma} Q_{N, T}(\mathbf{g}, \boldsymbol{\mu})$.
In practice, the estimator is computed by the iterative \emph{kmeans} procedure. We start with an initial group membership structure $\mathbf{g}^{(0)}$ and then iterate $\boldsymbol{\mu}$ and $\mathbf{g}$ such that the $s$-th iteration sets $\boldsymbol{\mu}^{(s)} = \arg \min_{ \boldsymbol{\mu} \in \mathcal{M}} Q_{N, T}(\mathbf{g}^{(s-1)}, \boldsymbol{\mu})$ and $\mathbf{g}^{(s)} = \arg \min_{ \boldsymbol{g} \in \Gamma} Q_{N, T}(\mathbf{g}, \boldsymbol{\mu}^{(s)})$ until convergence. Since the iteration may converge to a local minimum we re-start the procedure from many initial values for $\mathbf{g}$.

\section{Main results}
We consider asymptotic sequences under which $N, T \to \infty$ and 
\begin{align}
	\label{eq:asymptotic_sequence}
	\frac{(\log T) \sqrt{\log N}}{\sqrt{T}} = o(1).
\end{align}
We treat $(\sigma_1, \dotsc, \sigma_N)$ and $\mathbf{g}^0$ as unobserved deterministic parameters.

We first state sufficient conditions for consistent estimation of $\boldsymbol\mu^0$. 
\begin{assumption}
\label{as:minimal}
\begin{enumerate}[label=\roman*)]
\item 
\label{as:minimal:independence}
$\{v_{it}\}_{t=1}^T$ is an independent sequence with $\E v_{it} = 0$ and $\E v_{it}^2 = 1$. 
\item 
\label{as:minimal:average_variance}
The average error variance satisfies
% \begin{align*}
	$N^{-1} \sum_{i = 1}^N \sigma_i^2 = o (T)$.
% \end{align*}
\item 
\label{as:minimal:bounded_set}
There is a bounded set $\mathcal{M} \subset{\mathbb{R}}^G$ such that $\boldsymbol \mu^0 \in \mathcal{M}$.
\item 
\label{as:minimal:group_sep}
There is a positive constant $M_{G}$ such that 
\[
	\min_{g \in \mathbb{G}} \min_{h \in \mathbb{G} \setminus \{g\}} 
	\abs{\mu_g^0 - \mu^0_h} > M_G.
\]
\item 
\label{as:minimal:pi_min}
For all $g \in \mathbb{G}$,
$
		N^{-1} \sum_{i = 1}^N 1 (g^0_i = g) \geq q_{\min}
% \end{align*}
$.
\end{enumerate}
\end{assumption}
Part~\ref{as:minimal:independence} imposes independence of the error term over time. 
Using this assumption we obtain asymptotic results under simply conditions on between-unit heteroscedasticity. 
The assumption can be relaxed to allow for weak serial correlation at the expense of conditions on heteroscedasticity that are more difficult to interpret. 
Part~\ref{as:minimal:average_variance} states that the average error variance increases at a slower rate than $T$. This assumption ensures that, as $T \to \infty$, the additional information from observing more time periods is not undone by an increased noisiness of the signal. 
Part~\ref{as:minimal:bounded_set} is a standard regularity assumption.
Part~\ref{as:minimal:group_sep} requires that the group-specific means are distinct (group separation).  
Part~\ref{as:minimal:pi_min} ensures that the effective sample size that can be used to estimate the group-specific mean grows at the same asymptotic rate for all groups.

Assumption~\ref{as:minimal} does not restrict cross-sectional dependence. Assumption~\ref{as:normality} below limits the amount of cross-sectional dependence and is required for our result on $NT$-convergence of the group-specific parameters, but not any of our intermediate results. 

The grouped model is invariant to a relabeling of the groups and the vector of group-specific means $\boldsymbol\mu^0$ is therefore only identified up to a re-ordering of its components. The following result states that the identified set is consistently estimated. 

\begin{lemma}[Consistency of group-specific means]
\label{lem:mu_consistency}
	Suppose that Assumption~\ref{as:minimal} holds. Then, there is a (possibly random) permutation function $\pi: \mathbb{G} \to \mathbb{G}$ such that for all $\epsilon > 0$
	\begin{align*}
		\lim_{N, T \to \infty} P \left( \max_{g \in \mathbb{G}} \abs{\hat{\mu}_{\pi (g)} - \mu_g^0} > \epsilon \right) = 0.
	\end{align*}
\end{lemma}

Similarly to related results in the literature \parencite[e.g.][]{bonhomme2015grouped}, proving this result does not require establishing that group memberships are consistently estimated for all units. 
In Theorem~\ref{th:mu-consistent} below, we strengthen the result to root $NT$ convergence under weaker assumptions on heteroscedasticity than are commonly assumed in the literature.

The subsets of units for which we can guarantee that group memberships are uniformly consistently estimated is given by 
\begin{align}
\label{eq:I_N_T}
	\mathcal{I}_{N, T} = \left\{ 
		i \in \{1, \dotsc, N\} : \sigma_i \leq \frac{M_G}{140} \sqrt{\frac{T}{\log N}}
	\right\}.  
\end{align}
For the units in $\mathcal{I}_{N,T}$ the error variances are allowed to diverge but only at rate $\sqrt{T/\log N}$. Controlling the rate of divergence is necessary to ensure that observing additional time periods adds enough information to estimate group memberships precisely. What rates of divergence are permissible is determined by bounds on the tail of the error distribution. 
The error term of our panel model is given by $\sigma_i v_{it}$. 
We assume that $v_{it}$ is a sub-exponential random variable.
Under this assumption, new observations add information at the usual parametric rate root $T$ and the price of uniformity is root $\log N$. 

\begin{assumption}[Sub-exponential errors]
\label{as:subexponential}
There are positive constants $\nu, \alpha$ such that 
\begin{align*}
	\max_{1 \leq i \leq N} \max_{1 \leq t \leq T} \E \exp(\lambda \abs{v_{it}}) \leq \exp\left(\frac{\lambda^2 \nu^2}{2}\right) 
	\qquad \text{for all $\lambda > 0$ such that $\lambda < \frac{1}{\alpha}$}.
\end{align*}
\end{assumption}
In addition to errors that are Gaussian and sub-Gaussian (conditional on $\sigma_i$) this assumption allows also for certain ``fat-tailed'' distributions such as Poisson or chi-squared. It is possible to relax this assumption and allow for distributions with even heavier tails, but only at the expense of a different rate condition in \eqref{eq:misclassified_rate} that is more difficult to state and to interpret. 
In our setting, misclassification can occur even for moderate realizations of $v_{it}$ if $\sigma_i$ is sufficiently large. Therefore, misclassification does not hinge on heavy tails of $v_{it}$ and is not ruled out or limited by Assumption~\ref{as:subexponential}. 

The following lemma states that group membership is estimated consistently uniformly over all units in $\mathcal{I}_{N,T}$.
\begin{lemma}
\label{lem:gi_consistency}
Suppose that Assumptions~\ref{as:minimal} and \ref{as:subexponential} hold. 
Then, 
there exists a (possibly random) permutation function $\pi: \mathbb{G} \to \mathbb{G}$ such that 
\begin{align*}
\lim_{N, T \to \infty} P \left( 
\sup_{i \in \mathcal{I}_{N, T}} 
\abs{\pi(\hat{g}_i) - g_i^0 } > 0
\right) \to 0.
\end{align*}
\end{lemma}
This lemma extends existing results in the literature that are derived under the assumption that $\max_{1 \leq i \leq N} \sigma_i^2$ is bounded in which case $\mathcal{I}_{N,T} = \{1, \dotsc, N\}$ eventually. 
Lemma~\ref{lem:gi_consistency} shows that uniform consistency over all units can be obtained even if the error variance $\sigma_i^2$ diverges for some or all units. In this case, all unit-specific error variances must diverge at most at the rate given in \eqref{eq:I_N_T} and the average error variance must diverge at most at the rate given in Assumption~\ref{as:minimal}\ref{as:minimal:average_variance}.

We study the asymptotic behavior of $\hat{\boldsymbol{\mu}}$ without requiring that all units are contained in $\mathcal{I}_{N,T}$ and therefore guaranteed to be estimated consistently. 
The idea of Theorem~\ref{th:mu-consistent} below is that units that are not in $\mathcal{I}_{N,T}$ do not affect the asymptotic distribution provided that there are sufficiently few of them. 

Let $\mathcal{I}_{N, T}^\mathsf{c} = \{1, \dotsc, N\} \setminus \mathcal{I}_{N, T}$ and write $\# A$ to denote the cardinality of a set $A$. 
We assume 
\begin{align}
	\label{eq:misclassified_rate}
	\frac{\# \mathcal{I}_{N,T}^{\mathsf{c}}}{N} 
	\max \left\{
		\sqrt{N T} , \sqrt{ N
			\frac{1}{\# \mathcal{I}_{N,T}^{\mathsf{c}}} \sum_{i \in \# \mathcal{I}_{N,T}^{\mathsf{c}}} \sigma_i^2}
		\right\}
	= o (1).
\end{align}
Existing theoretical results cover only settings under which no units are potentially misclassified in the asymptotic limit, i.e.,  $\# \mathcal{I}_{N,T}^{\mathsf{c}} = 0$. In this case \eqref{eq:misclassified_rate} is trivially satisfied.
Our result allows $\# \mathcal{I}_{N,T}^{\mathsf{c}} \neq 0$ provided that the proportion of possibly misclassified units ${\# \mathcal{I}_{N,T}^{\mathsf{c}}}/{N}$ vanishes at a sufficiently fast rate. The rate in the first component of the max ensures that units in $\mathcal{I}_{N,T}^{\mathsf{c}}$ asymptotically do not affect the mean of $\hat{\boldsymbol{\mu}}$. The rate of the second component in the max ensures that units in $\mathcal{I}_{N,T}^{\mathsf{c}}$ asymptotically do not affect the variance of $\hat{\boldsymbol{\mu}}$. By \eqref{eq:I_N_T}, the second component satisfies 
\begin{align*}
\sqrt{ N
			\frac{1}{\# \mathcal{I}_{N,T}^{\mathsf{c}}} \sum_{i \in \# \mathcal{I}_{N,T}^{\mathsf{c}}} \sigma_i^2}
> \sqrt{NT} \frac{ M_G }{140 \sqrt{\log N}}.
\end{align*}
This shows that the first component can dominate the second component at most at a root $\log N$ rate. Therefore, replacing the max in \eqref{eq:misclassified_rate} by the second component gives a good approximation (up to order root $\log N$) of the required rate condition.

To state the assumption for asymptotic normality of $\hat{\mu}_g$, $g \in \mathbb{G}$, let $\mathcal{I}_{N, T}(g) = \left\{i \in \mathcal{I}_{N, T} : g_i^0 = g \right\}$ and 
\begin{gather*}
	\tilde{N}_g =  \# \{ i \in \mathcal{I}_{N,T} : g_i^0 = g\}, \quad
	N_g =  \# \{ i \in 1, \dotsc, N: g_i^0 = g\}, \quad 
	\hat{N}_g =  \# \{ i \in 1, \dotsc, N : \hat{g}_i = g\}.
\end{gather*}

\begin{assumption} 
	\label{as:normality}
\begin{enumerate}[label = \roman*)]
	\item 
	\label{as:normality:inconsistent}
	Condition~\eqref{eq:misclassified_rate} is satisfied.
	\item 
	\label{as:normality:delta}
	 For each $g \in \mathbb{G}$ there are positive constants $\delta_g$ and $q_g$ such that $N_g/N \to q_g$ and  
\begin{align*}
	\frac{1}{\tilde{N}_g} \sum_{i \in \mathcal{I}_{N, T} (g)} \sigma_i^2 
	+ 
	\frac{1}{\tilde{N}_g} \sum_{\substack{i, j \in \mathcal{I}_{N, T} (g) \\i \neq j}} \sigma_i \sigma_j \cov(v_{i1}, v_{j1}) 
	\to \delta_g.
\end{align*}
	\item 
	\label{as:sigma_Lp_norm}
	We have 
	\begin{align*}
		\frac{1}{\# \mathcal{I}_{N,T}} \sum_{i \in \mathcal{I}_{N,T}} \sigma_i^2 = O (\sqrt{T})
		\quad \text{and} \quad 
		\frac{1}{\# \mathcal{I}_{N,T}} \sum_{i \in \mathcal{I}_{N,T}} \sigma_i^4 = O (N T).
	\end{align*}
	\item 
	\label{as:normality:limit-cs-depend}
	In addition, 
	\begin{align*}
		\sum_{\substack{i, j, k \in \mathcal{I}_{N, T}\\ \{i\} \cap \{j\} \cap \{k\} = \emptyset}}
		\sigma_i \sigma_j \sigma_k 
		\E [v_{i1}^2 v_{j1} v_{k1}]
		= & O (N^2 T), 
		\\
		\sum_{\substack{i, j, k, \ell \in \mathcal{I}_{N, T} \\ \{i\} \cap \{j\} \cap \{k\} \cap \{\ell\}= \emptyset}} 
		\sigma_i \sigma_j \sigma_k \sigma_\ell
		\E [v_{i1} v_{j1} v_{k1} v_{\ell 1}]
		= & O (N^2 T).
	\end{align*}
\end{enumerate}
\end{assumption}
Part~\ref{as:normality:delta} ensures that the asymptotic variance of $\hat{\mu}_g$ converges.
Part~\ref{as:sigma_Lp_norm} imposes two conditions on the rate of divergence of the $L_2$ and the $L_4$ norm of $\{\sigma_i : i \in \mathcal{I}_{N,T}\}$. Under cross-sectional independence the first condition is implied by \ref{as:normality:delta}. The second condition is satisfied if 
$
	N \log^2 N /T  \to \infty.
$
Part~\ref{as:normality:limit-cs-depend} limits the amount of cross-sectional dependence.

The following theorem guarantees root $NT$-consistency and asymptotic normality of $\hat{\mu}_g$.
\begin{theorem}
\label{th:mu-consistent}	
Suppose that Assumptions~\ref{as:minimal}--\ref{as:normality} hold. 
Then, for $g \in \mathcal{\mathbb{G}}$ as $N,T \to \infty$
\begin{align*} 
	\sqrt{N T} \left( \hat{\mu}_{\pi(g)} - \mu_g^0 \right) 
	\overset{d}{\longrightarrow} \mathcal{N} (0,  q_g^{-1} \delta_g) .
\end{align*}
\end{theorem}

This result shows that root $NT$-consistency can be obtained even if some units are potentially misclassified in the limit. 
In addition, the error variance for the units that are consistently estimated need not be bounded. For root $NT$-consistency we require a stronger assumption on the average error variance than for the result on consistent estimation of group memberships in Lemma~\ref{lem:gi_consistency}. Assumption~\ref{as:normality}\ref{as:normality:delta} implies that the average error variance is bounded. In contrast, Lemma~\ref{lem:gi_consistency} allows the average error variance to diverge at a controlled rate.

\section{Conclusion}

We have shown that uniformly consistent estimation of group memberships is not a necessary condition of root $NT$ estimation of time invariant group-specific parameters. The simple model with group-specific intercepts served our purpose of providing an example of a grouped panel model in which a root $NT$ rate can be obtained even under misclassification in the limit.
We are confident that similar results can be obtained for general linear panel regression, albeit under more involved conditions that may not be straightforward to interpret. We leave such extensions to future research. For scenarios where the amount of misclassification permitted by our assumption \eqref{eq:misclassified_rate} is exceeded by only a sufficiently small margin, our proofs suggest that it is possible to obtain a convergence rate that is slower than root $NT$ but faster than root $N$. This suggests a negative relationship between the difficulty of classifying individual units and the precision of the estimator of the vector of group-specific coefficients.

\appendix

\section{Appendix: Mathematical proofs}

\begin{lemma}
\label{lem:subexponential}
Let $\mathcal{P}$ denote a class of probability measures that satisfy Assumption~\ref{as:subexponential}. Then
\begin{align*}
 \sup_{P \in \mathcal{P}} P \left(\max_{1 \leq i \leq N} \abs{\frac{1}{\sqrt{T}} \sum_{t = 1}^T v_{it}} >  14 \sqrt{\log N} \right) \leq 3 N^{-1}. 
\end{align*}
\end{lemma}
\begin{proof}
Fix a probability measure $P \in \mathcal{P}$ and let $\nu, \alpha >0$ denote the parameters from Assumption~\ref{as:subexponential}.
Let $\lambda^* > 0$ large enough that $\lambda^* < 1/\pi$ and $\exp (\nu^2 (\lambda^*)^2/ 2) \leq 2$. 
Define the Orlicz norm 
\begin{align*}
	\norm{v_{it}}_{\psi_1} = \inf \left\{\eta > 0 : \E \left[ \psi_1 (\abs{v_{it}}/\eta) \right] \leq 1 \right\} 
\end{align*}
with $\psi_1(t) = \exp(t) - 1$.
By Assumption~\ref{as:subexponential}, 
\begin{align*}
	\max_{1\leq i \leq N} \max_{1 \leq t \leq T} \E \exp (\lambda^* \abs{v_{it}}) \leq \exp\left(\nu^2 (\lambda^*)^2 / 2 \right) \leq 2.
\end{align*}
Defining $K = 1/\lambda^*$ this implies for all $1 \leq i \leq N$ and $1 \leq t \leq T$
\begin{align*}
	\E \left[
		\exp \left(\frac{\abs{v_{it}}}{K}\right) - 1
	\right] \leq 1
\end{align*}
and therefore
\begin{align*}
	\norm{v_{it}}_{\psi_1} = \inf \left\{\eta > 0 : \E \left[ \exp \left(\frac{\abs{v_{it}}}{\eta}\right) - 1 \right] \leq 1 \right\} \leq K.
\end{align*}
Hence, $\max_{1\leq i \leq N} \max_{1 \leq t \leq T} \norm{v_{it}}_{\psi_1} \leq K$. Applying Theorem~3.4 in \textcite{kuchibhotla2018moving} with $\alpha = 1$, $K_{n, q} = K$, $\Gamma_{n, q} = 1$ and $t = \log N$ yields 
\begin{align*}
	P \left( \max_{1 \leq i \leq N} \abs{\frac{1}{T} \sum_{t = 1}^T v_{it}} > 7 \sqrt {\frac{2 \log N}{T}} + \frac{C_1 K \log (2 T)(2 \log N)}{T} \right) \leq 3 N^{-1}.
\end{align*}
By Assumption~\ref{as:subexponential}, 
\begin{align*}
	\frac{C_1 K \log (2 T)(2 \sqrt{\log N})}{\sqrt{T}}  = o(1)
\end{align*}
and therefore 
\begin{align*}
	14 \sqrt{\frac{\log N}{T}} > 7 \sqrt {\frac{2 \log N}{T}} + \frac{C_1 K \log (2 T)(2 \log N)}{T}.
\end{align*}
\end{proof}
\begin{lemma}
\label{lem:Q_tilde}
Suppose that Assumption \ref{as:minimal}\ref{as:minimal:independence}--\ref{as:minimal:bounded_set} holds.
Then, for all $\epsilon > 0$
\begin{align*}
	\lim_{N, T \to \infty}  P \left( \sup_{\boldsymbol{g} \in \Gamma, \boldsymbol{\mu} \in \mathcal{M}}
	\abs{
		Q_{N, T}(\mathbf{g}, \boldsymbol{\mu}) - \frac{1}{NT} \sum_{i = 1}^N \sumtT
		u_{it}^2
	+
	\frac{1}{N} \sum_{i = 1}^N
	\left(
		\mu_{g_i^0}^0 - \mu_{g_i}
	\right)^2
	} > \epsilon
	\right)
	= 0. 
\end{align*}
\end{lemma}
\begin{proof}
This proof is very similar to the proof of Lemma~A.1 in \textcite{bonhomme2015grouped}.
Expanding $Q_{N, T}$ gives 
\begin{align*}
	Q_{N, T}(\mathbf{g}, \boldsymbol{\mu}) = &
	\frac{1}{NT} \sum_{i = 1}^N \sumtT
		u_{it}^2
	+
	\frac{1}{N} \sum_{i = 1}^N
	\left(
		\mu_{g_i^0}^0 - \mu_{g_i}
	\right)^2
\\
	& +
	\frac{2}{NT} \sum_{i = 1}^N \sum_{t = 1}^T \sigma_i v_{it} \left( \mu^0_{g_i^0} - \mu_{g_i} \right)
	. 
\end{align*}
By Cauchy-Schwarz
\begin{align*}
	\abs{\frac{1}{NT} \sum_{i = 1}^N \sumtT \sigma_i v_{it} \left( \mu^0_{g_i^0} - \mu_{g_i} \right)}^2
	\leq & C_{\mathcal{M}}
	\frac{1}{N} \sum_{i=1}^N \left\{\left( \frac{\sigma_i^2}{T}  \right) \left( \frac{1}{\sqrt{T}} \sumtT v_{it} \right)^2 \right\},
\end{align*}
where $C_{\mathcal{M}}$ is a constant that depends on a bound on $\mathcal{M}$.
Under the assumptions of the lemma,  
\begin{align*}
	 \frac{1}{N} \sum_{i=1}^N \E \left\{\left( \frac{\sigma_i^2}{T}  \right) \left( \frac{1}{\sqrt{T}} \sumtT v_{it} \right)^2 \right\} = o (1).
\end{align*}
Therefore, by Markov's inequality,  
\begin{align*}
 P \left( \frac{1}{N} \sum_{i=1}^N \left\{\left( \frac{\sigma_i^2}{T}  \right) \left( \frac{1}{\sqrt{T}} \sum_{t = 1}^T v_{it} \right)^2 \right\} > \epsilon\right)
= o (1).
\end{align*}
The conclusion follows.
\end{proof}
\begin{lemma}
\label{lem:muhat_MSEconv}
Suppose that Assumption \ref{as:minimal}\ref{as:minimal:independence}--\ref{as:minimal:bounded_set} holds. For each $\epsilon > 0$
	\begin{align*}
		\lim_{N, T \to \infty} P \left (\frac{1}{N} \sum_{i = 1}^N \left(
			\mu^0_{g_i^0} - \hat{\mu}_{\hat g_i}
		\right)^2 > \epsilon \right) = 0.
	\end{align*}
\end{lemma}
	\begin{proof}
	By definition,
	\begin{align*}
		Q_{N, T}(\hat{\mathbf{g}}, \hat{\boldsymbol{\mu}})
		\leq 
		Q_{N, T} (\mathbf{g}^0, \boldsymbol{\mu}^0).
	\end{align*}
	Let $W_{N, T}$ denote a random variable such that for each $\epsilon > 0$
	\begin{align*}
	\lim_{N, T \to \infty} P (\abs{W_{N, T}} > \epsilon) = 0.	
	\end{align*}
	Applying Lemma~\ref{lem:Q_tilde} to both sides of the inequality yields
	\begin{align*}
		\frac{1}{N} \sum_{i = 1}^N \left( \mu^0_{g_i^0} - \hat{\mu}_{\hat g_i} \right)^2
		\leq 
		\frac{1}{N} \sum_{i = 1}^N \left( \mu^0_{g_i^0} - \mu^0_{g_i^0} \right)^2 + W_{N, T}
	\end{align*}
	and the conclusion follows.	
\end{proof}
\begin{proof}[Proof of Lemma~\ref{lem:mu_consistency}]
This proof is very similar to the proof of Lemma~B.3 in \textcite{bonhomme2015grouped}.
	By Lemma~\ref{lem:muhat_MSEconv}
	\begin{align*}
		\frac{1}{N} \sum_{i = 1}^N \left(
			\mu^0_{g_i^0} - \hat{\mu}_{\hat g_i}
		\right)^2 = o_p(1).
	\end{align*}
	Suppose that there is a constant $\epsilon > 0$ and $g \in \mathbb{G}$ such that for $N,T \to \infty$ satisfying \eqref{eq:asymptotic_sequence}
	\begin{align}
		\label{eq:hausdorff:contradict:cond1}
		\limsup_{N, T \to \infty} P \left( \min_{h \in \mathbb{G}} \abs{\hat{\mu}_h - \mu_{g}^0} > \frac{\epsilon}{q_{\min}} \right) \geq \epsilon.
	\end{align}
	Under $\min_{h} \abs{\hat{\mu}_h - \mu_{g}^0} > \epsilon / q_{\min}$ we have
	\begin{align*}
		\frac{1}{N} \sum_{i = 1}^N \left(
					\mu^0_{g_i^0} - \hat{\mu}_{\hat g_i}
				\right)^2 > 
		\frac{1}{N} \sum_{\substack{i = 1, \dotsc, N \\ g^0(i) = g}} \frac{\epsilon}{q_{\min}} \geq \epsilon
	\end{align*}
	and therefore 
	\begin{align*}
		\limsup_{N, T \to \infty} P \left(
		\frac{1}{N} \sum_{i = 1}^N \left(
					\mu^0_{g_i^0} - \hat{\mu}_{\hat g_i}
				\right)^2 > \epsilon
		\right) \geq \epsilon
		.
	\end{align*}
	This contradicts Lemma~\ref{lem:muhat_MSEconv}. Therefore \eqref{eq:hausdorff:contradict:cond1} does not hold and for all $\epsilon > 0$
	\begin{align*}
	% \label{eq:hausdorff:cond1}
		\lim_{N, T \to \infty} P \left(\max_{g \in \mathbb{G}} \min_{h \in \mathbb{G}} \abs{\hat{\mu}_h - \mu_{g}^0} > \epsilon \right)  
		\leq 
		\sum_{g \in \mathbb{G}} \lim_{N, T \to \infty} P \left(\min_{h \in \mathbb{G}} \abs{\hat{\mu}_h - \mu_{g}^0} > \epsilon \right) = 0.  
	\end{align*}
	This result implies that, for any constant $0 < \epsilon < M_G/2$ and	
	\begin{align*}
		\limsup_{N, T \to \infty} P \left (\max_{g \in \mathbb{G}} \min_{h \in \mathbb{G}} \abs{\hat{\mu}_h - \mu_g^0} \geq \epsilon \right) < \epsilon.
	\end{align*}	
	If 
	\begin{align*}
		\max_{g \in \mathbb{G}} \min_{h \in \mathbb{G}} \abs{\hat{\mu}_h - \mu_g^0} < \epsilon
	\end{align*}
	then there exists, to each $g \in \mathbb{G}$, a non-empty set $H_g \subset \mathbb{G}$ such that $\abs{\hat{\mu}_h - \mu^0_g} < \epsilon$ for all $h \in H_g$. We now prove $H_g \cap H_{g'} = \emptyset$ for $g, g' \in \mathbb{G}$ with $g \neq g'$. Suppose $h \in H_g$. Then 
	\begin{align*}
		\abs{\hat{\mu}_{h} - \mu^0_{g'}} = \abs{\hat{\mu}_{h} - \mu^0_{g} + \mu^0_{g} - \mu^0_{g'}} 
		\geq \abs{\mu^0_{g'} - \mu^0_g} - \abs{\hat{\mu}_{h} - \mu^0_{g}} \geq M_G - \epsilon > \epsilon.
	\end{align*}
	Therefore $h \neq H_{g'}$ and $H_{g} \cap H_{g'} = \emptyset$. Since $H_{g} \neq \emptyset$ this implies that all sets $H_g$, $g \in \mathbb{G}$ are singletons. Define the function $\pi: \mathbb{G} \to \mathbb{G}$ that maps each group $g$ to the unique $h$ such that $\abs{\hat{\mu}_h - \mu^0_g} < \epsilon$. The function $\pi$ is a bijection and hence a permutation function. For any given $h \in \mathbb{G}$ setting $g = \pi^{-1}(h)$ guarantees $\lvert{\hat{\mu}_h - \mu_{g}^0}\rvert < \epsilon$. 
	Therefore, 
	\begin{align*}
		\limsup_{N, T \to \infty} P\left (\max_{h \in \mathbb{G}} \abs{\hat{\mu}_{\pi(g)} - \mu_g^0} \geq \epsilon \right) \leq \epsilon.
	\end{align*}	
\end{proof}
\begin{proof}[Proof of Lemma \ref{lem:gi_consistency}]
Let $\pi: \mathbb{R} \to \mathbb{R}$ denote the permutation function from Lemma~\ref{lem:mu_consistency}. For $i = 1, \dotsc, N$, we have $\hat{g}_i \neq \pi(g_i^0)$ only if there is $g \in \mathbb{G} \setminus \{\pi(g_i^0) \}$ such that 
\begin{align*}
	& \sum_{t=1}^T \left( y_{it} - \hat{\mu}_{\pi (g_i^0)} \right)^2
	\geq
	\sum_{t=1}^T \left( y_{it} - \hat{\mu}_g \right)^2.
\end{align*}
Plugging in $y_{it} = \mu^0_{g_i^0} + \sigma_i v_{it}$ and rewriting the inequality yields
\begin{align*}
\sgn(\hat{\mu}_g - \hat{\mu}_{\pi (g_i^0)}) \frac{1}{\sqrt{T}  } \sum_{t = 1}^T v_{it} 
\geq \frac{\sqrt{T}}{2 \sigma_i} \abs{\hat{\mu}_g - \hat{\mu}_{\pi(g_i^0)}}
- 
\sgn(\hat{\mu}_g - \hat{\mu}_{\pi (g_i^0)}) \frac{\sqrt{T}}{\sigma_i} (\mu_{g_i^0} - \hat{\mu}_{\pi(g_i^0)}) 
.
\end{align*}
Let $\mathcal{E}_{N, T}$ denote the event 
\begin{align*}
	\mathcal{E}_{N, T} = \{\max_{g \in \mathbb{G}} \abs{\hat{\mu}_{\pi (g)} - \mu_g^0} > {M_G}/{5} \}.
\end{align*}
% By Lemma~\ref{lem:mu_consistency}, $\lim_{N, T \to \infty} \sup_{P \in \mathbb{P}} P \mathcal{E}_{N,T} = 0$.
On $\mathcal{E}_{N, T}$, 
\begin{align*}
	& \frac{\sqrt{T}}{2 \sigma_i} \abs{\hat{\mu}_g - \hat{\mu}_{\pi(g_i^0)}}
- 
\sgn(\hat{\mu}_g - \hat{\mu}_{\pi (g_i^0)}) \frac{\sqrt{T}}{2\sigma_i} (\mu_{g_i^0} - \hat{\mu}_{\pi(g_i^0)}) 
\\
\geq & \frac{\sqrt{T}}{2 \sigma_i} \left(
	\abs{\mu^0_{\pi^{-1}(g)} - \mu^0_{g_i^0}}
	- \abs{\hat{\mu}_g - \mu^0_{\pi^{-1}(g)}} 
	- 3 \abs{\hat{\mu}_{\pi(g_i^0)} - \mu^0_{g_i^0}}
\right) \geq \frac{\sqrt{T}}{10 \sigma_i} M_G.
\end{align*}
Therefore, 
\begin{align*}
	& P \left(\max_{i \in \mathcal{I}_{N,T}} \abs{\hat{g}_i - \pi (g_i^0)} > 0 \right) 
\\
	\leq &
	P \left( \text{there exists $i \in \mathcal{I}_{N,T}$ such that} \:  \sgn(\hat{\mu}_g - \hat{\mu}_{\pi (g_i^0)}) \frac{1}{\sqrt{T}  } \sum_{t = 1}^T v_{it}  
	\geq \frac{\sqrt{T}}{10 \sigma_i} M_G \right) + P \left(\mathcal{E}_{N,T} \right)
\\
	\leq &
	P \left( \max_{1 \leq i \leq N} \abs{\frac{1}{\sqrt{T}} \sum_{t = 1}^T v_{it}}
	\geq 14 \sqrt{\log N} \right) + P \left(\mathcal{E}_{N,T} \right), 
\end{align*}
where the last inequality follows since 
\begin{align*}
	\frac{\sqrt{T}}{10 \sigma_i} M_G \geq 14 \sqrt{\log N}
\end{align*}
for all $i \in \mathcal{I}_{N,T}$ and $\mathcal{I}_{N,T} \subset \{1, \dotsc, N\}$.
By Lemma~\ref{lem:mu_consistency} and Lemma~\ref{lem:subexponential}, 
\begin{align*}
	\lim_{N, T \to \infty} \Bigg[ P \left( \max_{1 \leq i \leq N} \abs{\frac{1}{\sqrt{T}} \sum_{t = 1}^T v_{it}}
	\geq 14 \sqrt{\log N} \right) + P \left(\mathcal{E}_{N,T} \right) \bigg] = 0.
\end{align*}
\end{proof}
\begin{proof}[Proof of Theorem \ref{th:mu-consistent}]
Throughout the proof we omit the $N, T$ subscripts and write $\mathcal{I}$, $\mathcal{I}(g)$ and $\mathcal{I}^\mathsf{c}$ instead of $\mathcal{I}_{N,T}$, $\mathcal{I}_{N,T} (g)$ and $\mathcal{I}^\mathsf{c}_{N,T}$.
Assumption~\ref{as:normality}\ref{as:normality:inconsistent} implies 
\begin{align*}
	\frac{\# \mathcal{I}_{N, T}^\mathsf{c}}{N} = o (1).
\end{align*}
Hence, for $g \in \mathbb{G}$,  
\begin{align*}
	1 \leq \frac{\tilde{N}_g}{N_g} \leq \frac{N_g}{N_g} + \frac{\# \mathcal{I}_{N, T}^\mathsf{c}}{N_g} 
	\leq 1 + (1 + o (1)) \frac{q_g \# \mathcal{I}_{N, T}^\mathsf{c}}{N} 
	\leq 1 + o (1)
\end{align*}
and therefore 
\begin{align*}
	\abs{\frac{\tilde{N}_g}{N_g} - 1} = o (1).
\end{align*}
For $g \in \mathbb{G}$,  
\begin{align*}
	\frac{\hat{N}_h}{\tilde{N}_g} = \frac{1}{\tilde{N}_g} \sum_{i \in \mathcal{I}^{\mathsf{c}}} 1(\pi (\hat{g}_i) \neq g) + \frac{1}{\tilde{N}_g} \sum_{i \in \mathcal{I}} 1 \left(\pi(\hat{g}) = g \right).
\end{align*}
By Lemma~\ref{lem:gi_consistency}
\begin{align*}
	\lim_{N, T \to \infty} P \left(
		\frac{1}{\tilde{N}_g} \sum_{i \in \mathcal{I}} 1 \left(\pi(\hat{g}) = g \right) \neq 1
	\right) = 0.
\end{align*}
Moreover, 
\begin{align*}
	\frac{1}{\tilde{N}_g} \sum_{i \in \mathcal{I}^{\mathsf{c}}} 1(\pi (\hat{g}_i) \neq g)
	\leq (1 + o (1)) \frac{\# \mathcal{I}_{N, T}^\mathsf{c}}{q_g N} \leq o (1)
\end{align*}
and therefore for all $\epsilon > 0$
\begin{align*}
	\lim_{N, T \to \infty} P \left(
		\abs{\frac{\hat{N}_h}{\tilde{N}_g} - 1} > \epsilon
	\right) = 0.
\end{align*}
For all $g \in \mathbb{G}$ we can bound 
\begin{align*}
	& \abs{
	\frac{1}{\hat{N}_g T} \sum_{i \in \mathcal{I}^{\mathsf{c}}} \sum_{t = 1}^T 1 \left(\pi(\hat{g}_i) = g \right) y_{it}} 
\\
	\leq &
	\frac{1}{q_g N} \left(1 + o_p(1) \right) \Bigg(  \sum_{i \in \mathcal{I}^{\mathsf{c}}}  1 \left(\pi(\hat{g}_i) = g \right) \abs{\mu_i^0} + \frac{1}{\sqrt{T}} \sum_{i \in \mathcal{I}^{\mathsf{c}}}  1 \left(\pi(\hat{g}_i) = g \right) \sigma_i \left(\frac{1}{\sqrt{T}} \sum_{t = 1}^T v_{it} \right)
	\Bigg)
\\
	\leq &
	\frac{1}{q_g N} \left(1 + o_p(1) \right) \Bigg( \# \mathcal{I}^{\mathsf{c}} \sup_{\boldsymbol \mu \in \mathcal{M}} \norm{\boldsymbol\mu}_{\max} + \frac{\# \mathcal{I}^{\mathsf{c}} }{\sqrt{T}} \sqrt{\frac{1}{\# \mathcal{I}^{\mathsf{c}}} \sum_{i \in \mathcal{I}^{\mathsf{c}}} \sigma_i^2} 
	\sqrt{\frac{1}{\# \mathcal{I}^{\mathsf{c}}} \sum_{i \in \mathcal{I}^{\mathsf{c}}}\left(\frac{1}{\sqrt{T}} \sum_{t = 1}^T v_{it} \right)^2}
	\Bigg)
	,
\end{align*}
where $\norm{\cdot}_{\max}$ is the max norm in $\mathbb{R}^G$.
By independence over time and $\E v_{it}^2 = 1$ we have
\begin{align*}
	\E \frac{1}{\# \mathcal{I}^{\mathsf{c}}} \sum_{i \in \mathcal{I}^{\mathsf{c}}}\left(\frac{1}{\sqrt{T}} \sum_{t = 1}^T v_{it} \right)^2 = 1
\end{align*}
and hence by the Markov inequality
\begin{align*}
 	\frac{1}{\# \mathcal{I}^{\mathsf{c}}} \sum_{i \in \mathcal{I}^{\mathsf{c}}}\left(\frac{1}{\sqrt{T}} \sum_{t = 1}^T v_{it} \right)^2 = O_p (1).
\end{align*} 
In addition, $ \sup_{\boldsymbol \mu \in \mathcal{M}} \norm{\boldsymbol\mu}_{\max}$ is bounded by Assumption~\ref{as:minimal}\ref{as:minimal:bounded_set}. 
Therefore 
\begin{align}
\label{eq:muhat:decompose:part1}
\begin{aligned}
	& \abs{
	\frac{1}{\hat{N}_g T} \sum_{i \in \mathcal{I}^{\mathsf{c}}} \sum_{t = 1}^T 1 \left(\pi(\hat{g}_i) = g \right) y_{it}} 
\\
	\leq &
	O(1) \left(1 + o_p(1) \right) \frac{\# \mathcal{I}^{\mathsf{c}}}{N} \left( 
		1 + (1 + O_p(1)) T^{-1/2} \sqrt{\frac{1}{\# \mathcal{I}^{\mathsf{c}}} \sum_{i \in \mathcal{I}^{\mathsf{c}}} \sigma_i^2} 
	\right) = o_p \left( \frac{1}{\sqrt{NT}} \right), 
\end{aligned}
\end{align}
where the last equality follows from Assumption~\ref{as:normality}\ref{as:normality:inconsistent}.
We will now apply the Lindeberg-Feller CLT to show
\begin{align}
\frac{1}{\sqrt{\tilde{N}_g}} \sum_{i \in \mathcal{I}(g)} 
\left\{\sigma_i \left(\frac{1}{\sqrt{T}} \sum_{t=1}^T v_{it} \right) \right\}
\overset{d}{\longrightarrow}
\mathcal{N}(0, \delta_g)
. \label{eq:lf-clt}
\end{align}
The variance of the term is given by 
\begin{align*}
 \E \left[
\frac{1}{{T}} \sum_{t=1}^T
\left(
\frac{1}{\sqrt{\tilde{N}_g}} \sum_{i \in \mathcal{I}(g)} 
\sigma_i v_{it} \right)^2
\right]
=  \frac{1}{\tilde{N}_g} \sum_{i \in \mathcal{I}(g)} 
\sigma_i^2 + 
\frac{1}{\tilde{N}_g} \sum_{\substack{i,j \in \mathcal{I}(g)\\i \neq j}} 
\sigma_i \sigma_j \cov(v_{i1}, v_{j1})
\to \delta_g
\end{align*}
To verify the Lindeberg condition it suffices to show that 
\begin{align}
	\label{eq:Lindeberg_sufficient}
	\E \left[ T^{-1/2} \sumtT z_{N, t}\right]^4 \leq K
\end{align}
eventually, where 
\[
	z_{N, t} = \frac{1}{\sqrt{\tilde{N}_g}} \sum_{i \in \mathcal{I}(g)} 
\sigma_i v_{it}
\]
and $K$ is a constant that does not depend on $N$ and $T$.
By independence across time periods 
\begin{align*}
	\E \left[ \frac{1}{\sqrt{T}} \sumtT z_{N, t}\right]^4
	= 
	\frac{\binom{4}{2}}{2!} \frac{1}{T^2}  \sum_{s=1}^T \sum_{t \neq s} \E [z_{N, s}^2]\E [z_{N, t}^2] 
	+ \frac{1}{T^2} \sum_{t = 1}^T \E [z_{N, t}^4] 
	=
	3 \delta_g^2 + \frac{1}{T^2} \sum_{t = 1}^T \E [z_{N, t}^4] + o \left( 1 \right)
	.
\end{align*}
To bound the right-hand side write for $t = 1, \dotsc, T$
\begin{align*}
	& \E \left[\sqrt{\tilde{N}_g} z_{N, t} \right]^4
	=
	\E \left[ \sum_{i \in \mathcal{I}(g)} \sigma_i v_{it} \right]^4
	\\
	= & 
	\sum_{i \in \mathcal{I}(g)} \sigma_i^4 \E [v_{it}^4] + 
	\frac{\binom{4}{2}}{2!}
	\sum_{\substack{i, j \in \mathcal{I}(g) \\ i \neq j}} \sigma_i^2 \sigma_j^2 \E [v_{it}^2 v_{jt}^2] 
	 + 
	\frac{\binom{4}{2}}{2!}
	\sum_{\substack{i, j, k \in \mathcal{I} \\ \{i\} \cap \{j\} \cap \{k\} = \emptyset}} 
	\sigma_i^2 \sigma_j \sigma_k  \E [v_{it}^{2} v_{jt} v_{kt}]
	\\
	& +
	\sum_{\substack{i, j, k, \ell \in \mathcal{I} \\ \{i\} \cap \{j\} \cap \{k\} \cap \{\ell\}= \emptyset}} 
	\sigma_i \sigma_j \sigma_k \sigma_\ell \E [v_{it} v_{jt} v_{kt} v_{\ell t}]
	= I_{1, t} + I_{2, t} + I_{3, t} + I_{4, t}.	
\end{align*}
To show \eqref{eq:Lindeberg_sufficient} it suffices to show $\sum_{t = 1}^T I_{k, t} = O (N^2 T^2)$ for $k = 1, \dotsc, 4$.
Assumption~\ref{as:normality}\ref{as:normality:inconsistent} implies 
\begin{align*}
	\abs{\frac{\# \mathcal{I}}{N} - 1} = o (1).
\end{align*}
Moreover, by Assumption~\ref{as:subexponential} there is a finite constant $M_4$ independent of $N$ and $T$ such that $\max_{1\leq t \leq T} \E [v_{it}^4] \leq M_4$. 
Therefore, 
\begin{align*}
	\sum_{t = 1}^T I_{1, t} \leq M_4 N T (1 + o (1))
	\left( \frac{1}{\# \mathcal{I}} \sum_{i \in \mathcal{I}} \sigma_i^4 \right) = 
	O (N^2 T^2).
\end{align*}
and
\begin{align*}
	\sum_{t = 1}^T I_{2, T} 
	\leq & 3 M_4 (1 + o(1) (N^2 T) 
	\left\{
	\frac{1}{\# \mathcal{I}} \sum_{i \in \mathcal{I}} \sigma_i^2 = O (N^2 T^2)
	\right\}^2.
\end{align*}
Moreover, Assumption~\ref{as:normality}\ref{as:normality:limit-cs-depend} yields $\sum_{t=1}^T I_{k, t} = O (N^2 T^2)$ for $k = 1, 2$. This proves \eqref{eq:Lindeberg_sufficient}.
For $g \in \mathbb{G}$ 
\begin{align*}
	\hat{\mu}_{\pi(g)} = &
	\frac{1}{\hat{N}_g T} \sum_{i \in \mathcal{I}^{\mathsf{c}}} \sum_{t = 1}^T 1 \left(\pi(\hat{g}_i) = g \right) y_{it}
	+ \frac{1}{\hat{N}_g T} \sum_{i \in \mathcal{I} \setminus \mathcal{I} (g)} \sum_{t = 1}^T 1 \left(\pi(\hat{g}_i) = g \right) y_{it}
\\
	& 	+ (1 + o_p(1)) \frac{1}{\tilde{N}_g T} \sum_{i \in \mathcal{I} (g)} \sum_{t = 1}^T 1 \left(\pi(\hat{g}_i) = g \right) (\mu_{g_i^0} + \sigma_i v_{it})
\end{align*}
The first term on the right-hand side is $o_p \left((NT)^{-1/2}\right)$ by \eqref{eq:muhat:decompose:part1}. 
The second term is $o_p \left((NT)^{-1/2}\right)$ by Lemma~\ref{lem:gi_consistency}. The third term converges to a centered normal with variance $\delta_g$ by \eqref{eq:lf-clt} and Slutzky's lemma.
\end{proof}

\printbibliography
	
\end{document}